\let\oldvec\vec
\documentclass[runningheads]{llncs}
\let\vec\oldvec
\usepackage{graphicx,amsmath,amsfonts,amssymb}
\usepackage{cite}

\newcommand{\E}{\mathbb{E}}

\newtheorem{prob}[definition]{Problem}

\begin{document}
\title{Online Payment Network Design}
%
%
\author{Georgia Avarikioti\and
Kenan Besic \and
Yuyi Wang \and Roger Wattenhofer}
\authorrunning{G. Avarikioti et al.}
%
\institute{ETH Z{\"u}rich, Switzerland
\\
\email{\{zetavar,besick,yuwang,wattenhofer\}@ethz.ch}}
\maketitle              
\begin{abstract}
Payment channels allow transactions between participants of the blockchain to be executed securely off-chain, and thus provide a promising solution for the scalability problem of popular blockchains. 
We study the online network design problem for payment channels, assuming a central coordinator. We focus on a single channel, where the coordinator desires to maximize the number of accepted transactions under given capital constraints. Despite the simplicity of the problem, we present a flurry of impossibility results, both for deterministic and randomized algorithms against adaptive as well as oblivious adversaries.


\end{abstract}
%
%
%
\section{Introduction}

Recently, blockchain systems and cryptocurrencies such as Bitcoin \cite{Nakamoto_bitcoin:a} and Ethereum \cite{buterin2013whitepaper} have gained in popularity -- in research, in economy and even in the general public. With the increased number of transactions, scalability has become a serious problem \cite{decker2015duplex,croman2016scaling}. The maximum  transaction throughput on the bitcoin network is approximately ten transactions per second. Other networks can process up to hundreds of transactions per second. In contrast, current digital payment systems (e.g. credit cards, WeChat Pay, etc.) handle tens of thousands of transactions per second \cite{visaTransactions}. 

Several solutions have been proposed to address the limitation on the transaction throughput on blockchain systems. Sharding \cite{kokoris2018omniledger,Luu2016ASS} is a so-called first-layer (on-chain) solution. 
However, the most promising approach are payment channels \cite{spilman2013channels,poon2015lightning,decker2015duplex}.
Payment channels allow transactions between two parties of a blockchain system to be executed off-chain. Furthermore, the existence of multiple two-party channels leads to the creation of payment networks, where transactions between a sender and a receiver can be executed through a path of channels in the network, even if the two parties have no direct channel with each other. Payment networks~\cite{poon2015lightning,decker2015duplex,miller2017sprites,Dziembowski2017PERUNVP,green2017bolt,decker2018eltoo,gudgeon2019offchain} operate \textit{on top} of the blockchain, introducing a second layer, and are thus known as \textit{Layer 2} solutions.

Even though payment networks are efficient and enable high transaction throughput on blockchain systems, they also demand from users to lock a lot of capital a priori.
In addition, complex routing algorithms are needed to discover routes with enough capital capacity from sender to receiver.
To address these issues, multiple proposals emerged that suggest the use of a central operator, in theory \cite{avarikioti2018algorithmic,avarikioti2018payment} and in practice \cite{Khalil2018nocust,poon2017plasma}.  
Concretely, \cite{avarikioti2018algorithmic} studied the complexity of a Payment Service Provider (PSP) from an algorithmic perspective. Similarly, \cite{avarikioti2018payment} examined the optimal graph structure and fee assignment to maximize a PSP's profit.

Similarly to \cite{avarikioti2018algorithmic,avarikioti2018payment}, we study the problem from the perspective of a PSP. The PSP creates the payment network and opens the channels between interacting parties in which the PSP locks the required capital, acting as a creditor. Since opening a channel is associated with cost (registration with the blockchain), the PSP charges fees to the customers using its channels. Additionally, the PSP decides which channels to open and how the capital will be distributed on the channels. The PSP's objective is to maximise its profit under capital constraints.

However, both  \cite{avarikioti2018payment} and \cite{avarikioti2018algorithmic} assumed to know all future transactions. While such assumptions may be valid in many situations, our paper takes a different route. We want to know how a PSP should set up channels \textit{without any assumption of future payments}. We wonder to what extent a PSP can still do a good job, and in which cases planning ahead is hopeless. In other words, we study the so-called \textit{online} version of setting up payments channels with a single PSP.

\subsection*{Our Contributions}
We study the single channel case, in which the PSP is called to decide which transactions to accept to maximize its profit on a single channel given capital constraints. 
First, we show there is no randomized online algorithm against any adaptive online adversary. Then, we consider algorithms against oblivious adversaries; we show that there is neither a competitive deterministic nor a competitive randomized algorithm. We derive our results for the randomized case from analysing deterministic algorithms with advice. Next, we consider resource augmentation \cite{Kalyanasundaram:2000:SPC:347476.347479}, an approach that relaxes the capital constraints to achieve better competitiveness. However, we prove there is no competitive deterministic algorithm, even with double the capital. 
Furthermore, we approach the problem as a minimization problem, where we want to minimize the number of rejected transactions. Similarly to the maximization problem, we prove there is no competitive randomized algorithm against oblivious adversaries.

\section{The Model}
In this section, we begin with a brief introduction on the payment channel operation and the current payment network design rationale. Then, we present the graph theoretic model for the payment network design problem, as addressed in this paper. Later, we define the problem variants and lastly, we present the different adversarial models we consider in this work.

\subsection{Payment Channels and Networks}
A payment channel is a construction that is established between two participants of a blockchain system and allows them to interact off-chain while maintaining the security guarantees of the blockchain. Thus, a payment network enables the exchange of capital without committing every transaction on the blockchain, and therefore addresses successfully the scalability problem of the underline blockchain.

To establish a payment channel, the participants publish a funding transaction, which is in essence an on-chain joint account where the participants of the channel lock their capital. As long as the funding transaction is securely included in the blockchain, the channel parties can execute transactions safely off-chain by updating the state of the channel, i.e., the distribution of capital between the participants. The new state of the channel is described in an update transaction which is signed by all the channel parties.

Multiple channels form a \emph{payment channel network}. To enable the transaction execution between two nodes that are not directly connected in the network, routing the transaction along a path of directly connected nodes is allowed. In such a case, we demand atomic execution of the transaction; either all payments in the path will be executed or none. 
To guarantee the atomic execution of transactions, Hashed Timelocked Contracts (HTLCs) \cite{decker2015duplex,poon2015lightning} can be used.


Routing in a payment network alleviates the necessity for direct links and hence the need for multiple channels (transactions). However, the update transactions in a payment channel are executed in private thus it is impossible to know the current distribution of the capital on a channel. This in turn leads to complex routing schemes, because often the chosen route is depleted. 

In our setting, we assume that a central authority, a Payment Service Provider (PSP), can open channels between two nodes and therefore knows how much capital is on each edge and how it is distributed between the nodes. The nodes are the PSP's customers. The PSP, as the creator of the network, bares the costs of opening channels. Further, we assume the PSP acts as a creditor, and hence locks the necessary capital in the network and then periodically gets paid by the customers (either in crypto or fiat money).
In our model, we assume, without loss of generality, that the cost for opening (or updating) a channel is $1$, which means that the total cost is the total number of funding transactions and channel updates in the network. Therefore, the fee that the PSP can charge a customer cannot be greater than $1$ per transaction. Otherwise, the customer would not execute the transaction through the payment network, but would instead directly use the blockchain.

\subsection{Graph Model}
We adapt the model of Avarikioti et al.~\cite{avarikioti2018algorithmic}.
We define the network as an undirected graph $G=(V,E)$ with a set of edges $E$ and a set of nodes $V$. A node $v \in V$ is a participant in the network whereas an edge $e = (u,v) \in E$ is a channel between two nodes $u,v \in V$. 
For each channel $(u,v)\in E$, we denote by $C_l$ and $C_r$ the capital available to node $u$ and node $v$, respectively.
Therefore, every time a transaction is executed through the channel $(u,v)$, the capitals $C_l$ and $C_r$ are updated according to the distribution of the capital in the latest update transaction of the channel. Note, that the total capital of a channel does not change, i.e., the sum of the capital $C_l$ and $C_r$ remains the same. Further, the capital moves on the channel like the balls in a row of an abacus: if $u$ wants to send  capital $c$ to $v$ on channel $(u,v)$ then the new distribution of the capital on the edge $(u,v)$, after the off-chain execution of the transaction, will be $C_l-c$ and $C_r+c$, respectively.

We denote by $t=(s,r,v)$ a transaction $t$ that moves capital of value $v$ from sender $s\in V$ to receiver $r\in V$. Furthermore, in the simple case where we examine only a single channel, i.e., the graph has a single edge, we denote by $\langle C_l ; C_r \rangle$ the capital distribution on the single edge.


\subsection{Problem Variants}

In this section, we present the problem variants, as originally introduced in \cite{avarikioti2018algorithmic}.

\begin{prob}
[General Payment Network Design]
\label{General Payment Network Design}\hfill

\noindent\textsc{Input:} Capital $C$, profit $P$, the sequence of $n$ transactions $t_i = (s_i, r_i , v_i)$ with $1 \leq i \leq n$, each containing the sender node $s_i$, the receiver node $r_i$ and the value $v_i$ of the transaction $t_i$.\\
\textsc{Output:} Strategy $S = \{0, 1\}^n$, a binary vector where the $i^{th}$ position is $1$ if we choose to execute the $i^{th}$ transaction of the input and $0$ else. The graph $G(V, E, C_l , C_r)$ is the network we created to execute the chosen transactions, where $V$ is the set of senders and receivers that participate in any transaction, $E$ is the set of channels we open and $C_l$, $C_r$ the capital on each side of each edge. Each transaction can be routed arbitrarily in $G$, denoted by $S_e = \{-1, 0, 1\}^n$, for all $e \in E$, i.e., $S_e(i) = 1$ (or $-1$) if transaction $i$ is routed through edge $e$ from left to right (from right to left, respectively) and $S_e(i) = 0$ if transaction $i$ is not routed through edge $e$.

Our goal is to return (if it exists) a strategy $S$, a graph $G$ and a routing $S_e$ subject to the following constraints:
\begin{enumerate}
\renewcommand\labelenumi{\bfseries\theenumi}
    \item $|S|-|E|\geq P$
    \item $\forall{e\in E},\forall{j\in \{1,2,\dots,n\}}, -C_l(e) \leq \sum_{i=1}^{j}S_e(i) v_i \leq C_r(e)$
    \item $\sum_{e \in E}C_l(e) + C_r(e) + |E| \leq C$
\end{enumerate}
\end{prob}

\begin{prob}[Single Channel]
\label{Single Channel}
Given a sequence of $n$ transactions $t_i = (s, r, v_i)$, where $s$ and $r$ are the nodes of the single edge $e$, a capital assignment $C_r(e), C_l(e)$, and a profit $P$, decide
whether there is a strategy $S$ such that $|S|\geq P$ and $\forall{j} \in [n], -C_l(e) \leq \sum^{j}_{i=1} S(i)v_i \leq C_r(e)$.
\end{prob}


Problem \ref{General Payment Network Design} is general and, as shown in \cite{avarikioti2018algorithmic}, difficult to tackle even in the offline setting where the future transactions are known apriori. 
For this purpose, Problems \ref{Single Channel} was defined, as a subcase of the general problem that is of interest in our setting. We study this problem in the online setting, where we assume no prior knowledge about the future transactions.

\subsection{Adversary Models}
In literature~\cite{Ben-David90onthe}, there are three types of adversaries concerning online problems: the oblivious adversary, the adaptive online adversary and the adaptive offline adversary. 
These adversarial models are equally powerful regarding deterministic algorithms but may yield drastically different results when considering randomized algorithms. We define the adversarial models below with respect to our setting.

\begin{definition}[Oblivious adversary]
An oblivious adversary provides a sequence of transactions before an online algorithm starts its computations.
\end{definition}

\begin{definition}[Adaptive online adversary]\label{def:adaptonline}
The adaptive online adversary provides the next transaction based on the decision an online algorithm makes (accept or reject the previous transaction) but serves it immediately.
\end{definition}

\begin{definition}[Adaptive offline adversary]\label{def:adaptofline}
The adaptive offline adversary provides the next transaction based on the decision an online algorithm makes (accept or reject the previous transaction) and serves the output at the end such that it acts optimally.
\end{definition}

Note that an adaptive offline adversary (Definition \ref{def:adaptofline}) knows the randomness of the online algorithm, in contrast to the adaptive online adversary (Definition \ref{def:adaptonline}).

\section{Single Channel}
In this section, we consider the online version of Problem \ref{Single Channel}. In other words, we study whether there is a profitable strategy for accepting and rejecting transactions (online) given capital constraints.

To this end, we define the competitive ratio of an algorithm, denoted by $c$, as the value that upper bounds the ratio between the profit of the optimal offline solution, $OPT(I_x)$, and the  profit of the online algorithm, $ALG(I_x)$, for any input sequence $I_x$, i.e.,
$$\forall x,~c \geq \frac{OPT(I_x)}{ALG(I_x)}$$
The profit, in both cases, represents the number of accepted transactions.

The rest of the section is structured as follows:
First, we discuss the existence of competitive algorithms against adaptive adversaries. Later, we mainly focus on algorithms, deterministic or randomized, against oblivious adversaries. In particular, we show lower bounds on the advice bits for competitive deterministic algorithms with advice. Then, we use the relationship between advice and randomization to discuss the competitiveness of randomized online algorithms. Last, we consider resource augmentation, i.e. we assume an online algorithm has more resources (more capital on both sides) than an optimal offline algorithm.

\subsection{Randomized algorithms against adaptive adversaries}\label{subsec:random-adaptive}
As shown in recent previous work \cite{avarikioti2018algorithmic}, there is no competitive deterministic algorithm against adaptive adversaries.
In this work, we extend this result and prove there is no competitive randomized algorithm against adaptive adversaries.

\begin{theorem}
    There is no competitive randomized algorithm against adaptive online adversaries.
\end{theorem}
\begin{proof}
	Assume that there is a competitive randomized algorithm against any adaptive online adversary.
	Theorem 2.1 and Theorem 2.2 in \cite{Ben-David90onthe} state that the existence of a randomized $c$-competitive algorithm against any adaptive online adversary and the existence of a randomized $d$-competitive algorithm against any oblivious adversary imply the existence of a $cd$-competitive randomized algorithm against any adaptive offline adversary. Additionally, this implies that there is a $cd$-competitive deterministic algorithm. Since a $c$-competitive algorithm against any adaptive online adversary is $c$-competitive against any oblivious adversary, the existence of a $c$-competitive randomized algorithm against any adaptive online adversary implies the existence of a $c^2$-competitive deterministic algorithm.
	From Theorem 19 in \cite{avarikioti2018algorithmic} we know that there are no deterministic online algorithms against adaptive online adversaries which contradicts our assumption and proves that there is no competitive randomized algorithm.\hfill \qed
\end{proof}

From here on, we discuss algorithms against oblivious adversaries since there are no algorithms against stronger adversaries. 



\subsection{Algorithms with advice}\label{subsec:advice}
In this subsection, we examine algorithms with advice.
We refer to an algorithm as \emph{optimal}, if the competitive ratio is one, hence it performs as well as the optimal offline algorithm. Next, we study optimal deterministic algorithms with advice, i.e. we examine how many advice bits are necessary for an optimal deterministic algorithm.
We prove a tight lower bound of $n-2$ advice bits for optimal deterministic algorithms. Later, we study the relation between the necessary number of advice bits and the competitiveness of an algorithm.

\begin{theorem}
    There exists an optimal deterministic algorithm with $n-2$ advice bits.
\end{theorem}
\begin{proof}
    The advice (oracle) gives the strategy for the first $n-2$ transactions. For the last two transactions, the algorithm proceeds greedily. If the optimal offline algorithm accepts none of the last two transactions, none of them can be accepted. If one or two transactions can be accepted, our algorithm will accept and reject them accordingly such that it is optimal.\hfill \qed
\end{proof}
    
\begin{theorem}
    There is no optimal deterministic algorithm with less than $n-2$ advice bits.
\end{theorem}
\begin{proof}
    Assume an algorithm $ALG$ reading at most $n-2$ advice bits exists. Then, $ALG$ reads no advice bits for $n=3$. We construct following sequences of transactions with initial capital distribution $\langle 2;1 \rangle $:
    \begin{enumerate}
        \item ($l,r,2$),($l,r,1$),($l,r,1$)
        \item ($l,r,2$),($r,l,3$),($r,l,3$)
    \end{enumerate}
    For the second sequence of transactions, the optimal offline algorithm accepts the first transaction, whereas, for the first sequence, the first transaction is rejected. Since $ALG$ is deterministic, the decision made on the first transaction will always be the same. Consequently, $ALG$ cannot be optimal.\hfill \qed
\end{proof}    

We notice that the number of advice bits necessary for optimal algorithms is very high. Therefore, we consider, next, algorithms with competitive ratio greater than one, and show a relation between the competitive ratio and the required advice bits of an algorithm.

\begin{theorem}
    An algorithm with strictly less than $f(n)$ advice bits has a competitive ratio of $\Omega(\frac{n}{f(n)})$.
\end{theorem}
\begin{proof}
    Let the initial capital distribution be $\langle 2^{f(n)};1 \rangle$ and $f(n)$ be a function such that $f(n) \in o(n)$. $loop$ denotes transactions that move the total capital from left to right until there are $n$ transactions. We define a sequence of transactions as $$I_x = (l,r,1), \dots, (l,r,2^{f(n)-1}), (r,l,x+1), loop$$ such that the set of instances is defined as $$\mathcal{I} = \{I_x \mid x \in \{0, \dots, 2^{f(n)}-1\}\}$$
    Similar to the proofs before, an algorithm must accept $loop$, otherwise the competitive ratio is $c \geq \frac{n-f(n)}{f(n)+1}$. Assume that there is an algorithm $A$ reading at most $k < f(n)$ advice bits. For instance $I_x$, $A$ must transfer $x$ capital with the transactions of the prefix to the right side such that it can accept $loop$. The number of strategies the advice can describe is smaller than the number of instances in the constructed input set. Thus, we know that for two different inputs the same strategy is used. We conclude, $A$ cannot accept $loop$ for one element in $\mathcal{I}$ and has competitive ratio $c \geq \frac{n-f(n)}{f(n)+1}$.\hfill \qed
\end{proof}

\subsection{Randomized algorithms against oblivious adversaries}\label{subsec:random-oblivious}
In this subsection, we study the competitive ratio of randomized algorithms. To that end, we use the results from section \ref{subsec:advice} to provide lower bounds on the competitive ratio of randomized algorithms, since algorithms with advice and randomized algorithms are closely related. According to  \cite{Borodin:1998:OCC:290169}, Yao's Principle \cite{yao1977probabilistic} bounds the expected competitive ratio of randomized algorithms from below with
$$R \geq max\left(min_i \frac{\E_{y(j)}[OPT(\sigma_j)]}{\E_{y(j)} [ALG_{i}(\sigma_j)]}, min_i \frac{1}{\E_{y(j)}[\frac{ALG_i(\sigma_j)}{OPT(\sigma_j)}]}\right)$$
where both arguments in the maximum function are proven lower bounds. $\E[X]$ is the expectation of a random variable $X$,  $OPT$ is the optimal offline algorithm and $ALG_i$ is the optimal strategy for the input sequence $\sigma_i$. Moreover, $y(j)$ denotes the distribution of $j$ and $\sigma_j$ is an input sequence for all $j$. We denote by $c$ the expected competitive ratio and by $I_x$ the input sequence. We change the other variables accordingly.

\begin{theorem}\label{thm:random-competitive}
    Every randomized algorithm is $\Omega(\frac{n}{f(n) + \frac{n}{2^{f(n)}}})$-competitive.
\end{theorem}
\begin{proof}
    Let $\langle 2^{f(n)};1 \rangle $ be the initial capital distribution and $$\mathcal{I} = \{I_x \mid x \in\ \{0,\dots,2^{f(n)}-1\}\}$$ the set of sequences of transactions for $f(n) \in o(n)$ where $$I_x = p, (r,l,x+1), loop$$ and $p = (l,r,1), (l,r,2), \dots, (l,r,2^{f(n)-1})$. $loop$ is the movement of the capital back and forth. Then, we define the set of strategies as $\mathcal{A} = \{A_x \mid x \in \{0, \dots, 2^{f(n)}-1\}\}$. 
    A non-optimal strategy accepts at most $f(n)$ transactions, whereas the optimal strategy accepts at most $n$ and at least $n-f(n)$ transactions. We will refer to this in the following calculation as $(*)$. 
    \begin{align*}
        c & \geq min_i \frac{1}{\E_{y(j)}[\frac{A_i(I_j)}{OPT(I_j)}]} & \mbox{Yao's principle} & \\
        & = \frac{1}{max_i \E_{y(j)}[\frac{A_i(I_j)}{OPT(I_j)}]}&\\
        & \geq \frac{1}{max_i \E_{y(j)}[\frac{A_i(I_j)}{n-f(n)}]} & \mbox{$OPT(I_j) \geq n-f(n)$}&\\
        & = \frac{n-f(n)}{max_i \E_{y(j)}[A_i(I_j)]}&\\
        & \geq \frac{n-f(n)}{\sum\limits_{k=0, k \neq i}^{2^{f(n)}-1} (\frac{1}{2^{f(n)}}(f(n)+1)) + \frac{n}{2^{f(n)}}}&\mbox{$y(j)$ uniform and $(*)$}&\\
        & = \frac{n-f(n)}{(2^{f(n)}-1)(\frac{f(n)+1}{2^{f(n)}}) + \frac{n}{2^{f(n)}}}&\\
        & = \frac{n-f(n)}{f(n) + 1 - \frac{f(n)+1}{2^{f(n)}} + \frac{n}{2^{f(n)}}}&\\
        & \geq \frac{n-f(n)}{2(f(n) + \frac{n}{2^{f(n)}})}&\\
    \end{align*}
    Thus, we have shown that $c \in \Omega(\frac{n}{f(n) + \frac{n}{2^{f(n)}}})$.\hfill \qed
\end{proof}

\begin{corollary}
There is no competitive randomized algorithm against oblivious adversaries.
\end{corollary}
\begin{proof}
Follows immediately from Theorem \ref{thm:random-competitive} for $f(n)=o(n)$.
\hfill \qed
\end{proof}

\subsection{Resource augmentation}\label{subsec:resource}

In this subsection, we approach the problem from another perspective, and consider a different analysis approach called resource augmentation as described in  \cite{Komm:2016:IOC:3036138}.
In this approach, an online algorithm may have more resources than the optimal offline algorithm in order to improve the performance of an algorithm against an adversary.
In our setting, we allow the online algorithm to have $h \geq 1$ times more capital on both sides. An online algorithm starts with $\langle hC_l;hC_r \rangle $ if the optimal offline algorithm starts with the initial capital distribution $\langle C_l;C_r \rangle $. In this section, we discuss the existence of deterministic algorithms for different values of $h$.

\begin{theorem}
    There is no competitive deterministic algorithm for $h<2$.
\end{theorem}
\begin{proof}
    Let the initial capital distribution for the online algorithm be $\langle hC;0 \rangle = \langle 2C-\epsilon;0 \rangle$ such that $h = 2 - \frac{\epsilon}{C}$.
    Then, we define the sequences
    \begin{align*}
        &(l,r,C-\frac{\epsilon}{2}), (r,l,1), (r,l,1), \dots, (r,l,1) \\
        &(l,r,C-\frac{\epsilon}{2}), (l,r,C),(r,l,C), \dots, (l,r,C)
    \end{align*}
    Assume there is a competitive algorithm that rejects the first transaction. Then, none of the consecutive transactions with low value (first sequence) can be accepted. This contradicts the assumption and implies that if a competitive deterministic algorithm exists, the first transaction must be accepted.
    Suppose there is a competitive algorithm that accepts the first transaction. Then, the capital distribution after accepting the first transaction is $\langle C-\frac{\epsilon}{2};C-\frac{\epsilon}{2} \rangle$. Then, the following transactions in the second sequence cannot be accepted, since the the movement of $C$ back and forth is not possible anymore. This contradicts the assumption that there is a competitive deterministic algorithm that accepts the first transaction. We conclude, that there is no deterministic algorithm for $h = 2 - \frac{\epsilon}{C}$. For $C=n$, $C=2^n$, or even bigger $C$, $h$ goes to $2$ from below. Thus, there is no competitive deterministic online algorithm for any $h<2.$\hfill \qed
\end{proof}

\subsection{Minimizing the number of rejected transactions}
So far, we defined the competitiveness as the ratio between the number of accepted transactions of the optimal offline algorithm and an online algorithm, thus as an online maximisation problem. In this section, we examine the problem from another point of view; defining a minimization problem. 
To this end, we define the competitive ratio of an algorithm, denoted by $c$, as the value that upper bounds the ratio between the  cost of the online algorithm, $ALG(I_x)$, and the cost of the optimal offline solution, $OPT(I_x)$, for any input sequence $I_x$, i.e.,
$$\forall x,~c \geq \frac{ALG(I_x)}{OPT(I_x)}$$
The cost, in both cases, represents the number of rejected transactions.

Similarly to sections \ref{subsec:advice} and \ref{subsec:random-oblivious}, we first lower bound the competitive ratio for a given upper bound on the advice bits and then use this result to show there is no competitive randomized algorithm against oblivious adversaries.

\begin{theorem}\label{minRejectedAdvice}
    An algorithm with strictly less than $f(n)$ advice bits has a competitive ratio of $\Omega(\frac{n}{f(n)})$.
\end{theorem}
\begin{proof}
    Let $\langle 2^{f(n)};1\rangle$ be the initial capital distribution.
    We define $$\mathcal{I} = \{I_x \mid x \in \{0,\dots,2^{f(n)}-1\}\}$$ where $$I_x = (l,r,1), (l,r,2),\dots(l,r,2^{f(n)-1}), (r,l,x+1), loop$$ for $f(n) \in o(n)$ and $loop$ being a sequence of transactions moving all the capital from left to right back and forth. An algorithm that does not accept $loop$ has a competitive ratio of at least $c \geq \frac{n-f(n)-1}{f(n)}$. Assume there is an algorithm reading $k < f(n)$ advice bits with a better competitive ratio. Since there are $2^{f(n)}$ different instances in $\mathcal{I}$ and fewer strategies are expressible by the advice, one strategy is used for two elements of the input set. In the previous section, we showed that transferring $i$ is the only strategy accepting $loop$ for an instance $I_i$. Then, for $I_j$ where $j \neq i$ the strategy that transfers $i$ with the transactions of the prefix cannot accept $loop$. As a result, the competitive ratio is at least $c \geq \frac{n-f(n)-1}{f(n)}$ which contradicts the assumption that there is an algorithm with a better competitive ratio reading $k < f(n)$ advice bits. \hfill \qed
\end{proof}

As stated in \cite{Borodin:1998:OCC:290169}, Yao's Principle \cite{yao1977probabilistic} bounds the expected competitive ratio for minimisation problems from below with
$$\bar{R} \geq max\left(min_i \frac{\E_{y(j)} [ALG_{i}(\sigma_j)]}{\E_{y(j)}[OPT(\sigma_j)]}, min_i \E_{y(j)}[\frac{ALG_i(\sigma_j)}{OPT(\sigma_j)}]\right)$$
where both arguments in the maximum function are proven lower bounds. The notation is the same as in section \ref{subsec:random-oblivious}.

\begin{theorem}\label{thm:minimizing-reject}
    Every randomised algorithm (minimising the number of rejected transactions) is $\Omega(\frac{n}{f(n)})$-competitive.
\end{theorem}
\begin{proof}
    We construct the same set of instances as in the proof of Theorem \ref{minRejectedAdvice} with the same initial capital distribution. Accordingly, we define the set of strategies to be $\mathcal{A} = \{A_x \mid x \in \{0,\dots,2^{f(n)}-1\}\}$ where $A_x$ is the optimal strategy for $I_x$. As discussed in previous proofs, choosing a non-optimal strategy for an input results in the rejection of at least all of the $n-f(n)-1$ transactions in $loop$. We know that the optimal offline algorithm rejects at most $f(n)$ and at least no transactions.
    
    \begin{align*}
        c & \geq \min_i \frac{\E_{y(j)}[A_i(I_j)]}{\E_{y(j)}[OPT(I_j)]}&\mbox{Yao's Principle}&\\
        & \geq \min_i \frac{\E_{y(j)}[A_i(I_j)]}{f(n)}&\\
        & = \min_i \frac{\sum_{j=0,j \neq i}^{2^{f(n)}-1} ( A_i(I_j)) + OPT(I_j)}{2^{f(n)}f(n)} & \mbox{$y(j)$ uniform} &\\
        & \geq \frac{(2^{f(n)}-1)(n-f(n)-1)}{2^{f(n)}f(n)}&\\
        & = \frac{(1-\frac{1}{2^{f(n)}})(n-f(n)-1)}{f(n)}&\\
        & = (\frac{n}{f(n)}-1-\frac{1}{f(n)})(1-\frac{1}{2^{f(n)}}) \in \Omega(\frac{n}{f(n)})&\\
    \end{align*}\hfill \qed
\end{proof}

\begin{corollary}
    There is no competitive randomised algorithm (that minimises the number of rejected transactions) against oblivious adversaries.
\end{corollary}

\begin{proof}
    Follows from Theorem \ref{thm:minimizing-reject} for $f(n)=1$. \hfill \qed
\end{proof}

\section{Related Work}

This work builds on the definitions and results of Avarikioti et al.~\cite{avarikioti2018algorithmic}, where a framework to approach the design of payment channel networks from an algorithmic perspective was originally introduced. For a single channel, they showed that maximising the profit with given capital assignments is NP-hard and presented a fully polynomial time approximation scheme in the offline setting, i.e. when the sequence of future transactions is known upfront. Moreover, they studied the online case, where no prior information is known about the future transactions. In particular, they showed that there is no competitive (deterministic) online algorithm and presented an $O(\log(C))$-competitive algorithm that constructs a payment hub that accepts all transactions.
In this paper, we extend their work for the online setting and consider randomized algorithms, algorithms with advice and resource augmentation algorithms. 

Additionally, the design of payment networks with fees from the viewpoint of a payment service provider who wants to maximise the profit is discussed in \cite{avarikioti2018payment}. In the contrary to this work where we assume constant fee for every transaction, in \cite{avarikioti2018payment}, each channel requires a different fee, much like the tolls on a road network. Despite the different assumptions, both works share the same objective, to maximize the profit for the network operator (and designer).

Payment channels were originally introduced by Spilman \cite{spilman2013channels} as a solution for the limited transaction throughput of Bitcoin \cite{Nakamoto_bitcoin:a}. Spilman channels allowed two parties to transact off-chain as long as the direction of the capital movement was only in one direction. Later, bidirectional channels were introduced simultaneously by Poon et al.~\cite{poon2015lightning} and by Decker and Wattenhofer \cite{decker2015duplex}. Many recent constructions of payment or state channels have been proposed addressing different aspects and needs (e.g. state channels handle smart contracts) of various cryptocurrencies \cite{miller2017sprites,decker2018eltoo,coleman2018counterfactual,avarikioti2019brick}.
Although these works propose different constructions for payment (or state) channels, they all result in a decentralized payment network and thus require complex routing algorithms and high capital availability on behalf of the users of the network. Therefore, the results of this work apply to all these proposed payment channel solutions.

Multiple works exist that focus on the routing problem of payment networks.
Prihodko et al.~propose Flare \cite{prihodko2016flare}, a proposal for path discovery by gathering information about the Lightning network topology. However, Flare raised privacy concerns which were later addressed by SilentWhispers \cite{moreno2017silentwhispers} and SpeedyMurmurs \cite{roos2017settling}. In contrast with these works, we assume a central authority, a payment service provider, that designs the network, and thus has complete knowledge on the network structure and capital capacity of each channel. Our objective is to design the optimal network structure to maximize the profit for the service provider.

On a different direction, Dziembowski et al.~proposed Perun \cite{Dziembowski2017PERUNVP}, a virtual channel hub that allows the users connected to the hub to directly interact off-chain via virtual channels establish through the virtual hub. 
In the same line of work, Heilman et al.~presented Tumblebit \cite{heilman2017tumblebit}, a payment channel hub compatible with Bitcoin that guarantees anonymity and security even though the hub is an untrusted intermediary. 
Similarly, Green and Miers presented Bolt \cite{green2017bolt}, another channel construction that requires smart contracts but offers stronger privacy and security guarantees.
Moreover, Khalil et al.~introduced Nocust~\cite{Khalil2018nocust} whereas Poon and Buterin introduced Plasma~\cite{poon2017plasma}, which are layer-2 commit-chains, i.e., off-chain payment hubs.
Despite the fact that all these work also assume a central coordinator that enables the off-chain payments through a centralized network, they mainly focus on the construction of the payment hub.
In contrast, we focus on the algorithmic perspective of the problem, and address more primitive questions: how can the central coordinator profit the most, 
and which transactions should he facilitate through the network given a capital constraint.


\section{Conclusions}
We studied the online problem of a single channel, where a Payment Service Provider (PSP) is called to decide which transactions to accept or reject. The objective is to maximize the number or accepted transactions and thus maximize the PSP's profit. 

We showed there is no randomized online algorithm against any adaptive online adversary. Furthermore, we considered deterministic algorithms with advice and proved that there is no competitive randomized algorithm against an oblivious adversary. In addition, we examined resource augmentation, and showed that even with twice as much capital there is no competitive deterministic algorithm against an oblivious adversary. Finally, we considered the complementary minimization problem - minimizing the number of rejected transactions - and similarly proved there is no competitive randomized algorithm against oblivious adversaries.

Given that the single channel case is merely a simple sub-case of the general problem and the flurry of negative results we presented in this work, we conclude that the online channel design is a demanding problem with interesting future work.



%
%
%
\bibliographystyle{splncs04}
\bibliography{ref}

\begin{thebibliography}{10}
\providecommand{\url}[1]{\texttt{#1}}
\providecommand{\urlprefix}{URL }
\providecommand{\doi}[1]{https://doi.org/#1}

\bibitem{avarikioti2018payment}
Avarikioti, G., Janssen, G., Wang, Y., Wattenhofer, R.: Payment network design
  with fees. In: Data Privacy Management, Cryptocurrencies and Blockchain
  Technology, pp. 76--84. Springer (2018)

\bibitem{avarikioti2019brick}
Avarikioti, G., Kogias, E.K., Wattenhofer, R.: Brick: Asynchronous state
  channels (2019)

\bibitem{avarikioti2018algorithmic}
Avarikioti, G., Wang, Y., Wattenhofer, R.: {Algorithmic Channel Design}. In:
  {29th International Symposium on Algorithms and Computation (ISAAC), Jiaoxi,
  Yilan County, Taiwan} (December 2018)

\bibitem{Ben-David90onthe}
Ben-David, S., Borodin, A., Karp, R., Tardos, G., Wigderson, A.: On the power
  of randomization in online algorithms. In: Algorithmica. pp. 379--386 (1990)

\bibitem{Borodin:1998:OCC:290169}
Borodin, A., El-Yaniv, R.: Online Computation and Competitive Analysis.
  Cambridge University Press, New York, NY, USA (1998)

\bibitem{buterin2013whitepaper}
Buterin, V.: Ethereum: A next-generation smart contract and decentralized
  application platform. \url{https://github.com/ethereum/wiki/wiki/White-Paper}
  (2013)

\bibitem{coleman2018counterfactual}
Coleman, J., Horne, L., Xuanji, L.: Counterfactual: Generalized state channels
  (2018)

\bibitem{croman2016scaling}
Croman, K., Decker, C., Eyal, I., Gencer, A.E., Juels, A., Kosba, A., Miller,
  A., Saxena, P., Shi, E., Sirer, E.G., Song, D., Wattenhofer, R.: On scaling
  decentralized blockchains. In: International Conference on Financial
  Cryptography and Data Security. pp. 106--125. Springer (2016)

\bibitem{decker2018eltoo}
Decker, C., Russell, R., Osuntokun, O.: eltoo: A simple layer2 protocol for
  bitcoin (2018)

\bibitem{decker2015duplex}
Decker, C., Wattenhofer, R.: A fast and scalable payment network with bitcoin
  duplex micropayment channels. In: Symposium on Self-Stabilizing Systems. pp.
  3--18. Springer (2015)

\bibitem{Dziembowski2017PERUNVP}
Dziembowski, S., Eckey, L., Faust, S., Malinowski, D.: Perun: Virtual payment
  channels over cryptographic currencies. IACR Cryptology ePrint Archive
  \textbf{2017}, ~635 (2017)

\bibitem{green2017bolt}
Green, M., Miers, I.: Bolt: Anonymous payment channels for decentralized
  currencies (10 2017)

\bibitem{gudgeon2019offchain}
Gudgeon, L., Moreno{-}Sanchez, P., Roos, S., McCorry, P., Gervais, A.: Sok: Off
  the chain transactions. {IACR} Cryptology ePrint Archive  \textbf{2019}, ~360
  (2019), \url{https://eprint.iacr.org/2019/360}

\bibitem{heilman2017tumblebit}
Heilman, E., Alshenibr, L., Baldimtsi, F., Scafuro, A., Goldberg, S.:
  Tumblebit: An untrusted bitcoin-compatible anonymous payment hub. In: Network
  and Distributed System Security Symposium (2017)

\bibitem{Kalyanasundaram:2000:SPC:347476.347479}
Kalyanasundaram, B., Pruhs, K.: Speed is as powerful as clairvoyance. J. ACM
  \textbf{47}(4),  617--643 (Jul 2000)

\bibitem{Khalil2018nocust}
Khalil, R., Gervais, A., Felley, G.: Nocust - a securely scalable commit-chain.
  Cryptology ePrint Archive, Report 2018/642 (2018),
  \url{https://eprint.iacr.org/2018/642}

\bibitem{kokoris2018omniledger}
Kokoris-Kogias, E., Jovanovic, P., Gasser, L., Gailly, N., Syta, E., Ford, B.:
  Omniledger: A secure, scale-out, decentralized ledger via sharding. In: 2018
  IEEE Symposium on Security and Privacy (SP). pp. 583--598. IEEE (2018)

\bibitem{Komm:2016:IOC:3036138}
Komm, D.: An Introduction to Online Computation: Determinism, Randomization,
  Advice. Springer Publishing Company, Incorporated, 1st edn. (2016)

\bibitem{Luu2016ASS}
Luu, L., Narayanan, V., Zheng, C., Baweja, K., Gilbert, S., Saxena, P.: A
  secure sharding protocol for open blockchains. In: ACM Conference on Computer
  and Communications Security (2016)

\bibitem{miller2017sprites}
Miller, A., Bentov, I., Kumaresan, R., Cordi, C., McCorry, P.: Sprites and
  state channels: Payment networks that go faster than lightning. arXiv
  preprint arXiv:1702.05812  (2017)

\bibitem{moreno2017silentwhispers}
Moreno-Sanchez, P., Kate, A., Maffei, M.: Silentwhispers: Enforcing security
  and privacy in decentralized credit networks (2017)

\bibitem{Nakamoto_bitcoin:a}
Nakamoto, S.: Bitcoin: A peer-to-peer electronic cash system (2018),
  \url{http://bitcoin.org/bitcoin.pdf}

\bibitem{poon2017plasma}
Poon, J., Buterin, V.: Plasma: Scalable autonomous smart contracts (2017)

\bibitem{poon2015lightning}
Poon, J., Dryja, T.: The bitcoin lightning network: Scalable off-chain instant
  payments (2015), \url{https://lightning.network}

\bibitem{prihodko2016flare}
Prihodko, P., Zhigulin, S., Sahno, M., Ostrovskiy, A., Osuntokun, O.: Flare :
  An approach to routing in lightning network white paper (2016)

\bibitem{roos2017settling}
Roos, S., Moreno-Sanchez, P., Kate, A., Goldberg, I.: Settling payments fast
  and private: Efficient decentralized routing for path-based transactions.
  arXiv preprint arXiv:1709.05748  (2017)

\bibitem{spilman2013channels}
Spilman, J.: Anti dos for tx replacement. \\
  \url{https://lists.linuxfoundation.org/pipermail/bitcoin-dev/2013-April/002433.html},
  accessed: 2019-04-17

\bibitem{visaTransactions}
{Visa Inc.}: Fact sheet - visa. \\
  \url{https://usa.visa.com/dam/VCOM/download/corporate/media/visanet-technology/aboutvisafactsheet.pdf}
  (2018), acccessed: 10.04.2019

\bibitem{yao1977probabilistic}
Yao, A.C.C.: Probabilistic computations: Toward a unified measure of
  complexity. In: 18th Annual Symposium on Foundations of Computer Science
  (sfcs 1977). pp. 222--227. IEEE (1977)

\end{thebibliography}

\end{document}